\documentclass[11pt]{article}
\usepackage[top=1in,bottom=1in,left=1in,right=1in]{geometry}
\usepackage{amsmath,amssymb,amsthm}
\usepackage{xcolor}
\usepackage{algorithm}
\usepackage{algpseudocode}
\usepackage{hyperref}
\hypersetup{
  colorlinks=true,
  linkcolor=blue,
  citecolor=blue,
  urlcolor=blue,
  pdftitle={A Public-Key-Dependent Adversarial-Deletion Ceiling for Fixed-Alphabet Multi-Bit Pseudorandom Codes},
  pdfauthor={Frederick Dehmel and Venkat Guruswami and Shilun Li}
}

\theoremstyle{plain}
\newtheorem{theorem}{Theorem}
\newtheorem{lemma}[theorem]{Lemma}
\newtheorem{proposition}[theorem]{Proposition}
\newtheorem{corollary}[theorem]{Corollary}
\theoremstyle{definition}
\newtheorem{definition}[theorem]{Definition}

\newcommand{\Enc}{\mathrm{Encode}}
\newcommand{\Dec}{\mathrm{Decode}}
\newcommand{\KGen}{\mathrm{KeyGen}}
\newcommand{\pk}{\mathrm{pk}}
\newcommand{\sk}{\mathrm{sk}}
\newcommand{\LCS}{\mathrm{LCS}}
\newcommand{\bits}{\{0,1\}}
\newcommand{\Sigmaq}{\Sigma_q}
\newcommand{\Cpk}{\mathcal{C}^{\pk}_{\delta}}
\newcommand{\subseq}{\sqsubseteq}
\newcommand{\negl}{\mathrm{negl}}
\newcommand{\E}{\mathbb{E}}
\DeclareMathOperator{\Bin}{Bin}

\newcommand{\gm}[1]{\gamma_2^{(#1)}}

\DeclareMathOperator{\Geo}{Geom}
 \DeclareMathOperator{\poly}{poly}

\algrenewcommand\algorithmicrequire{\textbf{Input:}}
\algrenewcommand\algorithmicensure{\textbf{Output:}}
\algnewcommand{\LineComment}[1]{\State \(\triangleright\) \textit{#1}}

\definecolor{MyDarkBlue}{rgb}{0,0.08,1}
\definecolor{MyDarkGreen}{rgb}{0.02,0.6,0.02}
\definecolor{MyDarkRed}{rgb}{0.8,0.02,0.02}
\definecolor{MyDarkOrange}{rgb}{0.40,0.2,0.02}
\definecolor{MyPurple}{RGB}{111,0,255}
\definecolor{MyRed}{rgb}{1.0,0.0,0.0}
\definecolor{MyGold}{rgb}{0.75,0.6,0.12}
\definecolor{MyDarkgray}{rgb}{0.66, 0.66, 0.66}

\title{A Public-Key-Dependent Adversarial-Deletion Ceiling for\\ Fixed-Alphabet Multi-Bit Pseudorandom Codes}
\author{
Frederick Dehmel\thanks{Department of Electrical Engineering and Computer Sciences, UC Berkeley, Berkeley, CA 94720, USA. Email: \url{dehmelf@berkeley.edu}.} 
\and
Shilun Li\thanks{Department of Mathematics, UC Berkeley, Berkeley, CA, 94720, USA. Supported in part by DOD Advanced Research Projects Agency grant HR0011262E031. Email: \url{shilun@berkeley.edu}.}
}
\date{}

\begin{document}
\maketitle

\begin{abstract}
A pseudorandom code (PRC) is a keyed error-correcting code whose codewords are computationally indistinguishable from uniform strings. We study public-key PRCs over fixed alphabets against adversarial deletions, where the deletion channel may both depend on the public encoding key and the transmitted codeword.

Let \(\gamma_q^{\mathrm{LCS}}\) denote the asymptotic normalised longest-common-subsequence length of two independent uniform \(q\)-ary strings. We prove that for every fixed \(q\ge2\) no multi-message public-key PRC with a single-output decoder is robust against all such \(\delta\)-deletion channels for any $\delta>1-\gamma_q^{\mathrm{LCS}}$. For \(q=2\), the current rigorous bound \(\gamma_2\ge0.792665992\) rules out every constant \(\delta>0.207334008\). The proof uses pseudorandomness only to transfer an LCS event from uniform strings to independently sampled codewords and therefore also the resulting collision argument is information-theoretic and requires no secret key. We also extend the argument to list decoding: for every fixed constant $L$, provided the message space contains at least $L+1$ messages, no such PRC with output lists of size at most $L$ is robust for $\delta>1-\gamma_2^{(L+1)}$. Since \(\gamma_2^{(m)}=1/2+\Theta(1/\sqrt m)\), these thresholds approach \(1/2\). Our bounds are specific to public-key-dependent adversarial channels and do not apply to oblivious edit channels.

\end{abstract}

\section{Introduction}\label{sec:intro}

A \emph{pseudorandom error-correcting code} (PRC), as first introduced by Christ and
Gunn~\cite{ChristGunn}, is a keyed encoding scheme whose codewords are computationally indistinguishable from uniformly
random strings yet remain decodable with a secret key after passing through a noisy channel. PRCs underlie recent constructions for undetectable and robust watermarking
of the outputs of generative models~\cite{ChristGunn, GolowichMoitra, GunnZhaoSong}, where the pseudorandomness hides the
watermark and the error correction lets it survive edits.

The amount of noise that a PRC can tolerate is hence a central quantitative question. For the binary substitution channel this question has now been well studied~\cite{GGW,DMR}, whereas the \emph{deletion} channel is more difficult to study as removed symbols destroy positional synchronisation, rendering length-preserving substitution arguments inapplicable. We ask how large of a deletion fraction a fixed-alphabet PRC can withstand when the channel can access both the public key as well as the transmitted codeword, but not the secret key.

\subsection{Our results}
For every fixed alphabet size $q \ge2$ we prove that no multi-bit, undetectable public-key PRC with a single-valued decoder is robust against $\Cpk$ for any constant 
  \[
  \delta>1-\gamma_q^{\mathrm{LCS}},
  \]
  where $\gamma_q^{\mathrm{LCS}}$ is the asymptotic normalised LCS length of two
  uniform $q$-ary strings (Corollary~\ref{cor:qary}). In binary, current rigorous
  bounds yield impossibility for every $\delta>0.207334008$
  (Theorem~\ref{thm:main}).

The proof is organised around a \emph{Collision Lemma} (Lemma~\ref{lem:collision}). Suppose two independently sampled codewords have an LCS of length at least $\rho n$ with probability at least $1-\nu$. For every constant $\delta>1-\rho$, the lemma constructs two efficient $\delta$-deletion channels using only the public key. Each channel resamples a codeword for the other message and deletes its input to a canonical common subsequence. Coupling the resampled codeword with the other transmitted codeword makes the two outputs coincide while preserving each channel's standalone distribution. It follows that at least one of the corresponding decoding experiments succeeds with probability at most $\tfrac12+\tfrac{\nu}{2}$.

The same mechanism also extends to list decoding, as for every fixed constant $L\ge1$, provided the message space contains at least $L+1$ elements, no binary undetectable public-key PRC whose decoder outputs lists of size at most  $L$ is robust against $\Cpk$ for any constant
\[
\delta>1-\gm{L+1},
\]

where $\gm m$ is the asymptotic normalised LCS length of $m$ independent uniform binary strings (Theorem~\ref{thm:stair}). The attack maps $L+1$ codewords to a common subsequence which no list of size $L$ can decode to all $L+1$ messages, and since $\gm m=1/2+\Theta(1/\sqrt m)$~\cite{LiRenWen}, the resulting thresholds approach $1/2$ as $1/2-\Theta(1/\sqrt L)$.

\subsection{Related work}\label{sec:related}

PRCs were first introduced by Christ and Gunn~\cite{ChristGunn} who constructed LPN-based public-key zero-bit PRCs 
robust to a constant rate of substitutions and to i.i.d.\ deletions. Their syntax and security will form the basis for ours but we omit their standard soundness requirement because it is unused. Ghentiyala and Guruswami~\cite{GhentiyalaGuruswami} subsequently gave a public-key construction stemming from the planted-hyperloop assumption and a variant of Goldreich's local PRG, as
well as from weak planted XOR together with LPN.
Golowich and Moitra~\cite{GolowichMoitra} constructed edit-robust PRCs over alphabets that grow polynomially with the security parameter, which lies outside our fixed-alphabet setting. Christ, Golowich, Gunn, Moitra, and  Wichs~\cite{CGGMW} constructed secret-key, zero-bit binary PRCs robust to a constant fraction of worst-case edits. Their construction has neither the public sampler nor the distinct messages used by our collision attack.

Huang, Li, Mao, and Zhou~\cite{HLMZ} studied public-key PRCs for edit channels fixed outside the key-generation and encoding experiment. Their error patterns are chosen independently of the sampled keys and codeword and in particular they construct binary PRCs of rate $1/2-\epsilon$ robust against every sublinear-polynomial edit channel, together with higher-rate constructions over larger alphabets. These positive results are for oblivious channels and sublinear error, whereas we consider public-key-dependent channels with constant deletion fractions. Public-key-aware substitution attacks have been studied through \emph{adaptive robustness}. Cohen, Hoover, and Schoenbach~\cite{CHS} introduced adaptive-prompting robustness for watermarking, and Alrabiah, Ananth, Christ, Dodis, and Gunn~\cite{AACDG} formalised the corresponding PRC notion and proved adaptive substitution robustness up to $1/2-\epsilon$ for the Christ--Gunn zero-bit construction and up to $1/4-\epsilon$ for a single-bit variant. Their adversary selects the message and encoder randomness defining the reference codeword. A channel in $\Cpk$, however, is fixed before the experiment and acts on a freshly sampled codeword. Definition~\ref{def:adaptive-deletion} records the chosen-codeword deletion analogue, while our main theorem establishes the stronger fresh-sample impossibility.

The black-box separations of Garg, Gunn, and Wang~\cite{GGW} and Döttling, Müller, and Rajasree~\cite{DMR} apply to binary PRCs under length-preserving substitution noise, leveraging oracle separations and hypercontractivity. Our channels change length, and our separation is combinatorial and unconditional. Working under independent symbolwise corruption, Francati, Goonatilake, Pawar, Venturi, and Ateniese~\cite{FGPVA} derived an unconditional ceiling driven by the tension between soundness and robustness. Removing the assumption of soundness yields a distinct, non-overlapping threshold.

The LCS constant originates in the work of Chv\'atal and Sankoff~\cite{ChvatalSankoff}. For binary strings, the current rigorous bounds are
\[
0.792665992\le\gamma_2\le0.826280
\]
~\cite{HeinemanEtAl,Lueker}. Kiwi, Loebl, and Matou\v{s}ek~\cite{KLM} showed that
$\gamma_q^{\mathrm{LCS}}\sqrt q\to2$ as the alphabet size grows, while Li,
Ren, and Wen~\cite{LiRenWen} established $\gm m=1/2+\Theta(1/\sqrt m)$ for the
multiple-string constant underlying the list-decoding analysis.

In classical deletion coding, every sufficiently large binary codebook faces an
elementary $1/2$ list-decoding obstruction from monochromatic subsequences, and
Guruswami and Wang~\cite{GuruswamiWang} demonstrated the complementary
existence of positive-rate binary codes list-decodable from a $1/2-\epsilon$
fraction of deletions, and Bukh, Guruswami, and H{\aa}stad~\cite{BGH} connect
the pairwise constant $\gamma_2$ to the deletion radius of random binary
codes.

\section{Preliminaries and combinatorial input}\label{sec:prelim}

\subsection{Adversarial channel models}\label{sec:model}


Our impossibility result concerns the public-key-dependent deletion class $\Cpk$ of Definition~\ref{def:class}. A channel in this class receives the public key and transmitted word, may choose its deletion pattern as a function of both, and never accesses the secret key. The channel algorithm is fixed
before key generation and is applied to a freshly sampled codeword.

Two restricted cases will also be useful. A channel is \emph{key-independent} if it ignores the public-key argument, although it may still inspect the transmitted word. An \emph{oblivious} channel commits to its corruption
independently of both the sampled key and the transmitted word. 

Our attack uses $\pk$ only to run the public encoder and sample a partner codeword. It neither recovers the secret key nor chooses the randomness of the transmitted encoding which distinguishes the sampled-codeword model from the
chosen-codeword adaptive game of Definition~\ref{def:adaptive-deletion}, in which the adversary selects the message and encoder randomness defining the reference codeword.

\subsection{Definitions and notation}\label{sec:definitions}

Fix a constant alphabet size $q\ge2$, independent of $\lambda$, and let
$\Sigmaq:=\{0,1,\dots,q-1\}$. Deletions count symbols rather than bits in a
serialisation. Binary means $q=2$, when $\Sigma_2=\bits$.

For strings $a,b\in\Sigmaq^\ast$, write $a\subseq b$ when $a$ is a subsequence
of $b$, and let $\LCS(a,b)$ denote the length of their longest common
subsequence. For independent uniform $U,U'\in\Sigmaq^n$, define
\[
\gamma_q^{\mathrm{LCS}}
 :=\lim_{n\to\infty}\frac{\E[\LCS(U,U')]}{n}.
\]
The limit exists by superadditivity and Fekete's lemma
~\cite{ChvatalSankoff}. Write $\gamma_2:=\gamma_2^{\mathrm{LCS}}$. Any fixed
message used below denotes a family of the required length that is uniformly
computable from $1^\lambda$.

\begin{definition}[Pseudorandom code]\label{def:prc}
A \emph{$q$-ary pseudorandom error-correcting code} is a triple of PPT
algorithms $(\KGen,\Enc,\Dec)$ with polynomial-time-computable length functions
$n=n(\lambda)$ and $k=k(\lambda)$ satisfying
$1\le n(\lambda)\le\mathrm{poly}(\lambda)$ and
$0\le k(\lambda)\le\mathrm{poly}(\lambda)$:
\begin{itemize}
  \item $\KGen(1^\lambda)\to(\sk,\pk)$;
  \item $\Enc(\pk,m)\to\Sigmaq^n$ for $m\in\bits^k$;
  \item $\Dec(\sk,\cdot):\Sigmaq^\ast\to\bits^k\cup\{\bot\}$ is possibly
        randomised and single-valued.
\end{itemize}
The code is multi-bit if $k=k(\lambda)\ge 1$, so that at least two
distinct messages $m_0\neq m_1$ exist. (The case $k=0$, a single message, is
\emph{zero-bit}.)
\end{definition}

We omit the standard PRC soundness requirement because it is unused; hence our negative results also apply to PRCs satisfying it.

\begin{definition}[Pseudorandomness]\label{def:undet}
For every PPT distinguisher $A$ given $(1^\lambda,\pk)$ and access to an adaptive
chosen-message oracle, there is a negligible function
$\mu_A$ such that, for every security parameter $\lambda$,
\[
\left|
\Pr_{\substack{(\sk,\pk)\leftarrow\KGen(1^\lambda)\\ A,\Enc\text{-coins}}}
  [A^{\Enc(\pk,\cdot)}(1^\lambda,\pk)=1]
-\Pr_{\substack{(\sk,\pk)\leftarrow\KGen(1^\lambda)\\ A,\mathcal U\text{-coins}}}
  [A^{\mathcal U}(1^\lambda,\pk)=1]
\right|\le\mu_A(\lambda),
\]
where oracle queries lie in $\bits^{k(\lambda)}$, and the ideal oracle
$\mathcal U$ returns, on \emph{every} query (including repeats), a freshly and
independently drawn uniform value in $\Sigmaq^{n(\lambda)}$. The real oracle
likewise answers every query (including repeats) with an independent execution
of $\Enc(\pk,\cdot)$ on fresh coins.
\end{definition}

\begin{definition}[The public-key-dependent deletion class $\Cpk$]\label{def:class}
Fix a deletion fraction $\delta\in[0,1]$. A PPT machine $E$, taking as input
$(1^\lambda,\pk,y)$ and never reading $\sk$, belongs to $\Cpk$ iff,
for every $\lambda$, every $\pk$ in the support of the public-key component of
$\KGen(1^\lambda)$, and every $y\in\Sigmaq^{n(\lambda)}$:
\begin{enumerate}
  \item \textup{(deletion-only)} its output satisfies
        $E(1^\lambda,\pk,y)\subseq y$ with probability $1$ over $E$'s coins; and
  \item \textup{(budget)}
        $|E(1^\lambda,\pk,y)|\ge(1-\delta)n(\lambda)$ with probability $1$
        over $E$'s coins.
\end{enumerate}
The machine is uniform and total.
\end{definition}

\begin{definition}[Robustness]\label{def:robust}
The code is robust against $\Cpk$ if, for every $E\in\Cpk$, there is a
negligible function $\mu_E$ such that, for every $\lambda$ and every
$m\in\bits^{k(\lambda)}$,
\[
\Pr_{\substack{(\sk,\pk)\leftarrow\KGen(1^\lambda),\
x\leftarrow\Enc(\pk,m)\\E,\Dec\text{-coins}}}
\big[\Dec(\sk,E(1^\lambda,\pk,x))=m\big]
\ge1-\mu_E(\lambda).
\]
Here $\bot$ is a failure. The channel and message are fixed before key
generation, and the same $\mu_E$ applies to every message.
\end{definition}

\begin{definition}[Chosen-codeword adaptive deletion robustness]
\label{def:adaptive-deletion}
Fix a deletion fraction $\delta\in[0,1]$. The adaptive deletion experiment first
samples $(\sk,\pk)\leftarrow\KGen(1^\lambda)$ and gives
$(1^\lambda,\pk)$ to a PPT adversary $A$. The adversary outputs a message
$m\in\bits^{k(\lambda)}$, an encoder random tape $r$, and a string $z$. Set
$x:=\Enc(\pk,m;r)$. The output is \emph{admissible} when
\[
z\subseq x
\qquad\text{and}\qquad
|z|\ge(1-\delta)n(\lambda).
\]
The adversary wins when its output is admissible and
$\Dec(\sk,z)\ne m$, where the probability includes the decoder's coins. The
code is \emph{adaptively $\delta$-deletion robust} if every PPT adversary wins
with negligible probability. Thus, unlike the sampled-codeword experiment of
Definition~\ref{def:robust}, this adversary chooses the encoder randomness that
defines the reference codeword.
\end{definition}

\subsection{LCS concentration and transfer}\label{sec:combinatorial}

Fix a total order on $\Sigmaq^n$. Given $a,b\in\Sigmaq^n$,
canonically order the pair and run the standard LCS dynamic program with fixed
tie-breaking. Let $Z(a,b)$ be the resulting LCS and let
$J_a,J_b\subseteq[n]$ be the corresponding witnessing index sets. This rule
is deterministic, symmetric, and computable in $O(n^2)$ time, with
\[
(a_j)_{j\in J_a}=Z(a,b)=(b_j)_{j\in J_b},
\qquad
Z(a,b)=Z(b,a),
\qquad
|Z(a,b)|=\LCS(a,b).
\]

\begin{lemma}[Uniform lower tail]\label{lem:lowertail}
For every fixed rational $0<\rho<\gamma_2$, independent uniform
$U,U'\in\bits^n$, and all sufficiently large $n$,
\[
\Pr\big[\LCS(U,U')\ge\rho n\big]
\ \ge\ 1-\exp\!\big(-\tfrac14(\gamma_2-\rho)^2 n\big).
\]
Moreover, there is a constant $\beta_\rho>0$, depending only on $\rho$, such
that $\Pr[\LCS(U,U')\ge\rho n]\ge\beta_\rho$ for every $n\ge1$.
\end{lemma}
\begin{proof}
View $f:=\LCS(U,U')$ as a function of the $2n$ independent input bits.

\emph{$f$ is $1$-Lipschitz in each coordinate.} Fix all bits but one, say the
$j$-th bit of $U$, giving strings $a,a'\in\bits^n$ differing only at position
$j$, with $U'=b$ fixed. Let $S$ be a longest common subsequence of $a$ and $b$,
and fix witnessing index sets $I,J\subseteq[n]$ with
$(a_i)_{i\in I}=S=(b_j)_{j\in J}$. If $j\notin I$, then $S$ is itself a common
subsequence of $a'$ and $b$; otherwise deleting from $S$ the symbol matched at
position $j$ yields a common subsequence of $a'$ and $b$ of length
$\ge\LCS(a,b)-1$, so
$\LCS(a',b)\ge\LCS(a,b)-1$; by symmetry $|\LCS(a,b)-\LCS(a',b)|\le1$. The same
holds for any bit of $U'$. Hence each bounded-difference constant is $c_i=1$ and
$\sum_i c_i^2 = 2n$.

\emph{McDiarmid.} By the bounded-differences inequality, for every $t\ge0$,
\[
\Pr[f\le\E f - t]\ \le\ \exp\!\Big(-\frac{2t^2}{\sum_i c_i^2}\Big)
   \ =\ \exp\!\Big(-\frac{t^2}{n}\Big).
\]
Let $\alpha:=(\gamma_2-\rho)/2>0$. By the definition of $\gamma_2$, there is
$N_0$ such that $\E f\ge(\gamma_2-\alpha)n=(\rho+\alpha)n$ for all $n>N_0$.
Taking $t=\alpha n$, we have $\E f-t\ge\rho n$, and hence
\[
\Pr\big[\LCS(U,U')<\rho n\big]
 \ \le\ \exp(-\alpha^2 n)
 \ =\ \exp\!\big(-\tfrac14(\gamma_2-\rho)^2 n\big).
\]
Complementing yields the displayed bound. Enlarge $N_0$, if necessary, so that
the displayed probability is at least $1/2$ for $n\ge N_0$. For $n<N_0$, the
event $U=U'$ implies $\LCS(U,U')=n$ and has probability $2^{-n}$. Thus one may
take $\beta_\rho:=\min\{1/2,2^{-(N_0-1)}\}>0$. 
\end{proof}
\begin{lemma}[Transfer to real codewords]\label{lem:transfer}
Let $\mathcal C$ be an undetectable binary multi-bit public-key PRC, set
$m_0:=0^k$ and $m_1:=0^{k-1}1$, and draw
$x_i\leftarrow\Enc(\pk,m_i)$ independently under
$(\sk,\pk)\leftarrow\KGen(1^\lambda)$. For every fixed rational
$0<\rho<\gamma_2$ there exist a constant $N_\rho$ and a negligible function
$\mu_\rho$ such that, for every $\lambda$ with $n(\lambda)\ge N_\rho$,
\[
\Pr\big[\LCS(x_0,x_1)\ge\rho n\big]
\ \ge\ 1-\exp\!\big(-\tfrac14(\gamma_2-\rho)^2 n\big)-\mu_\rho(\lambda).
\]
For every block length, the same event has probability at least
$\beta_\rho-\negl(\lambda)$, where $\beta_\rho>0$ is from
Lemma~\ref{lem:lowertail}.
\end{lemma}
\begin{proof}
Write the fixed rational as $\rho=r/s$. The predicate
$P(a,b):=\mathbf1[s\LCS(a,b)\ge rn]$ is uniformly computable in
$O(n^2)=\mathrm{poly}(\lambda)$ time and does not require computing $\gamma_2$.
Consider the PPT distinguisher $A$ that, given $(1^\lambda,\pk)$ and an oracle
$\mathcal O$, computes $k(\lambda)$, constructs and queries $m_0$ and $m_1$,
receives $a,b$, and outputs $P(a,b)$. With the ideal
oracle the answers are independent uniform strings, so
Lemma~\ref{lem:lowertail} supplies both claimed lower bounds. With the real oracle
the answers are $x_0,x_1$. Pseudorandomness bounds the probability gap by
$\mu_A(\lambda)$. Take $N_\rho$ to be the threshold from
Lemma~\ref{lem:lowertail} and set $\mu_\rho:=\mu_A$. This proves both asserted
bounds and is the sole use of pseudorandomness in the main theorem.
\end{proof}

\section{Main collision lemma and deletion ceiling}\label{sec:collision}

\subsection{Collision lemma}

We isolate the core impossibility as a statement about \emph{any} fixed-alphabet
code with a public sampler and a single-valued decoder, parametrised by the
code's own pairwise-LCS rate. Neither a random-string LCS constant nor
pseudorandomness appears; both enter only in the specialisations below.

\begin{definition}[Public sampler]\label{def:sampler}
A \emph{public sampler} for a code is a PPT algorithm $S$ such that $S(\pk,m)$ is
distributed exactly as $\Enc(\pk,m)$. For a public-key code, $S:=\Enc$ is a
public sampler, since encoding is a public operation given $\pk$.
\end{definition}

\begin{lemma}[Collision Lemma]\label{lem:collision}
Let $\mathcal C=(\KGen,\Enc,\Dec)$ be a code with the syntax of
Definition~\ref{def:prc} (pseudorandomness not assumed), with single-valued
$\Dec$ and a public sampler $S$. Fix uniformly computable message families with
$m_0(\lambda)\neq m_1(\lambda)$ for every $\lambda$. Under
$(\sk,\pk)\leftarrow\KGen(1^\lambda)$, independently sample
$x_i\leftarrow\Enc(\pk,m_i)$ for $i\in\{0,1\}$. Suppose
\[
\Pr\big[\LCS(x_0,x_1)\ge\rho n\big]\ \ge\ 1-\nu
\]
for a fixed constant rate $0<\rho\le1$ and slack $\nu=\nu(\lambda)\in[0,1]$. Then for every constant deletion fraction $\delta\in(1-\rho,1]$ there exist channels $C_0,C_1\in\Cpk$ (each using no key material beyond $\pk$)
such that
\[
\min_{i\in\{0,1\}}\ \Pr\big[\Dec(\sk,C_i(1^\lambda,\pk,x_i))=m_i\big]
\ \le\ \tfrac12+\tfrac{\nu}{2}.
\]
\end{lemma}

\begin{proof}
\emph{Construction.} Fix the canonical LCS rule $Z$ from the preliminaries.
Choose a rational constant $\widehat\delta$ such that
$1-\rho<\widehat\delta<\delta$. Such a choice exists by density of the
rationals. The channels below enforce the stricter $\widehat\delta$ deletion
budget, and hence also belong to the $\delta$-deletion class $\Cpk$.
On input $1^\lambda$, both channels compute $k(\lambda)$ and the target-message
pair. They contain the constant $\widehat\delta$, receive $\pk$ and a string
$y$, and are total on every input; channel $C_0$ is
Algorithm~\ref{alg:C0}.

\begin{algorithm}[H]
\caption{$C_0(1^\lambda,\pk,y)$: channel attacking $m_0$}
\label{alg:C0}
\begin{algorithmic}[1]
\Require $1^\lambda$, public key $\pk$, and string $y$;\quad fixed: uniformly
computable partner family $m_1$ and rational $\widehat\delta$
\Ensure a subsequence of $y$ obtained using at most a
$\widehat\delta$-fraction of deletions
\State $N\gets|y|$; uniformly compute $m_1=m_1(\lambda)$ from $1^\lambda$
\State $p\gets S(\pk,m_1)$
\If{$|p|\ne N$}
    \State \Return $y_{1\,..\,\lceil(1-\widehat\delta)N\rceil}$
\EndIf
\State compute $z\gets Z(y,p)$ and canonical witnessing sets $J_y,J_p$
       such that $(y_j)_{j\in J_y}=z=(p_j)_{j\in J_p}$
\If{$|z|\ge(1-\widehat\delta)N$}
    \State \Return $\big(y_j\big)_{j\in J_y}$
\Else
    \State \Return $y_{1\,..\,\lceil(1-\widehat\delta)N\rceil}$
\EndIf
\end{algorithmic}
\end{algorithm}

Channel $C_1$ is defined symmetrically, using $m_0$ as its partner message.
Both channels are PPT, access no key material beyond $\pk$, and always output
either a witnessed common subsequence of length at least
$(1-\widehat\delta)N$ or a prefix of that length. Hence
$C_0,C_1\in\Cpk$.

For the coupled experiment, we sample
$(\sk,\pk)\leftarrow\KGen(1^\lambda)$ and independent
$x_i\leftarrow\Enc(\pk,m_i)$, and set $z^\ast:=Z(x_0,x_1)$. We then run $C_0$ using
$x_1$ as its partner and $C_1$ using $x_0$ and since $S(\pk,m_i)$ is distributed identically to $\Enc(\pk,m_i)$, each coupled invocation has its standalone
distribution. Symmetry of $Z$ ensures that, whenever the in-budget branch is
taken, both channels output the same string $z^\ast$.
Algorithm~\ref{alg:Estar} records this coupling.

\begin{algorithm}[H]
\caption{$E^{\ast}$: the coupled experiment realising the collision}
\label{alg:Estar}
\begin{algorithmic}[1]
\State $(\sk,\pk)\gets\KGen(1^{\lambda})$
\State $x_0\gets\Enc(\pk,m_0)$;\quad $x_1\gets\Enc(\pk,m_1)$ 
\State \textbf{couple:} run $C_0(1^\lambda,\pk,x_0)$ with partner $p:=x_1$, and
$C_1(1^\lambda,\pk,x_1)$ with $p:=x_0$
        
\State $z^{\ast}\gets Z(x_0,x_1)$
\State sample one decoder random tape $r$ and use it for both decoding events below
\LineComment{On $G=\{\LCS(x_0,x_1)\ge\rho n\}$ (probability $\ge1-\nu$ by hypothesis):}
\State \hspace{1em} by symmetry of $Z$,
$C_0(1^\lambda,\pk,x_0)=C_1(1^\lambda,\pk,x_1)=z^{\ast}$
\State \hspace{1em} for the shared tape, the two success events are disjoint
because $\Dec$ is single-valued and $m_0\neq m_1$
\end{algorithmic}
\end{algorithm}

For each $i$, the coupled triple
\[
(\sk,x_i,C_i(1^\lambda,\pk,x_i))
\]
is distributed exactly as in the standalone robustness experiment: $x_i$ has
the correct encoding distribution, and the channel's partner is an independent
sample from the required distribution. This identity does not use pseudorandomness. We share one uniform decoder tape $r$ between the two coupled
invocations so as to preserve each invocation's marginal
distribution, so
\[
p_i:=\Pr[\Dec(\sk,C_i(1^\lambda,\pk,x_i);r)=m_i]
\]
is its standalone success probability.

Let $G:=\{\LCS(x_0,x_1)\ge\rho n\}$, so
$\Pr[\overline G]\le\nu$. Because $\rho>1-\widehat\delta$, both channels take
the in-budget branch on $G$. The symmetry of $Z$ then gives
\[
C_0(1^\lambda,\pk,x_0)
=
Z(x_0,x_1)
=
Z(x_1,x_0)
=
C_1(1^\lambda,\pk,x_1)
=
z^\ast.
\]

\emph{Contradiction bound.} Let $D_i:=\{\Dec(\sk,z^\ast;r)=m_i\}$. On $G$,
$C_i(1^\lambda,\pk,x_i)=z^\ast$, so
\[
\Pr[D_i\cap G]\ =\
\Pr[\{\Dec(\sk,C_i(1^\lambda,\pk,x_i);r)=m_i\}\cap G]
\ \ge\ p_i-q,
\]
where $q:=\Pr[\overline G]\le\nu$.
For the shared tape $r$, $\Dec$ is single-valued; since $m_0\neq m_1$,
$D_0\cap D_1=\varnothing$, so
$D_0\cap G$ and $D_1\cap G$ are disjoint subsets of $G$:
\[
(p_0-q)+(p_1-q)\ \le\ \Pr[D_0\cap G]+\Pr[D_1\cap G]
\ \le\ \Pr[G]\ =\ 1-q.
\]
Hence $p_0+p_1\le1+q\le1+\nu$, so
$\min_i p_i\le\tfrac{1+\nu}{2}=\tfrac12+\tfrac\nu2$.
\end{proof}

\subsection{Binary adversarial-deletion ceiling}\label{sec:main}


\begin{theorem}[Binary multi-bit adversarial-deletion ceiling]\label{thm:main}
For every constant deletion fraction $\delta\in(1-\gamma_2,1]$, every binary,
multi-bit, unique-decoding (single-valued $\Dec$, Definition~\ref{def:prc}),
undetectable public-key PRC is \emph{not} robust against the
public-key-dependent $\delta$-deletion class $\Cpk$. In particular, the current
bound $\gamma_2\ge0.792665992$ proves the result for every
$\delta>0.207334008$. More quantitatively, after fixing any rational
$1-\delta<\rho<\gamma_2$, the two constructed pairs satisfy, pointwise for all
sufficiently large $\lambda$,
\[
\min_{i\in\{0,1\}}\Pr[\Dec(\sk,C_i(1^\lambda,\pk,x_i))=m_i]
\ \le\ 1-\frac{\beta_\rho}{2}+\negl(\lambda).
\]
If $n=\omega(\log\lambda)$, the right-hand side sharpens to
$\tfrac12+\negl(\lambda)$.
\end{theorem}
\begin{proof}
Fix a rational $1-\delta<\rho<\gamma_2$. By
Lemma~\ref{lem:transfer}, the Collision Lemma applies with
\[
\nu\le1-\beta_\rho+\negl(\lambda),
\]
and $\Enc$ is a public sampler. Hence
\[
\min_i\Pr[\Dec(\sk,C_i(1^\lambda,\pk,x_i))=m_i]
\le1-\frac{\beta_\rho}{2}+\negl(\lambda).
\]
Since there are only two channel--message pairs, one satisfies this bound for
infinitely many $\lambda$ and therefore violates robustness. If
$n=\omega(\log\lambda)$, Lemma~\ref{lem:transfer} instead gives
$\nu=\negl(\lambda)$, yielding the bound
$\tfrac12+\negl(\lambda)$.
\end{proof}

 \begin{corollary}[Fixed-alphabet deletion ceiling] \label{cor:qary}
Fix a constant alphabet size $q\ge2$. No multi-bit, undetectable public-key
PRC over $\Sigmaq$ with a single-valued decoder is robust against $\Cpk$ for
any constant
  \[
  \delta\in(1-\gamma_q^{\mathrm{LCS}},1].
  \]
\end{corollary}

\begin{proof}
Fix a rational $1-\delta<\rho<\gamma_q^{\mathrm{LCS}}$. The proof of Lemma~\ref{lem:lowertail} applies over $\Sigmaq$: the concentration argument is unchanged, while for the finitely many smaller block lengths the event $U=U'$ has probability $q^{-n}>0$. The rational-threshold predicate $\mathbf1[\LCS(a,b)\ge\rho n]$ is polynomial-time computable, so the transfer argument of Lemma~\ref{lem:transfer} also applies. The alphabet-agnostic Collision Lemma then gives the result.
\end{proof}

\begin{corollary}[Chosen-codeword adaptive deletion ceiling]
\label{cor:adaptive-deletion}
For every constant deletion fraction $\delta\in(1-\gamma_2,1]$, every binary,
multi-bit, single-valued, undetectable public-key PRC is not adaptively
$\delta$-deletion robust in the sense of
Definition~\ref{def:adaptive-deletion}. More quantitatively, after fixing any
rational $1-\delta<\rho<\gamma_2$, there is a PPT adversary whose winning
probability is at least
\[
\frac{\beta_\rho}{2}-\negl(\lambda),
\]
and this lower bound sharpens to $\tfrac12-\negl(\lambda)$ when
$n=\omega(\log\lambda)$.
\end{corollary}
\begin{proof}
On input $(1^\lambda,\pk)$, the adversary samples independent encoder tapes
$r_0,r_1$ and sets
$x_i:=\Enc(\pk,m_i;r_i)$ for the canonical distinct messages
$m_0=0^k$ and $m_1=0^{k-1}1$. It computes
$z^\ast:=Z(x_0,x_1)$ and samples $I\leftarrow\{0,1\}$ uniformly. If
$|z^\ast|\ge\rho n$, it outputs $(m_I,r_I,z^\ast)$; otherwise it outputs
$(m_I,r_I,x_I)$. The output is always admissible: the fallback is the identity,
while in the first branch $z^\ast\subseq x_I$ and
$|z^\ast|\ge\rho n>(1-\delta)n$.

Let $G:=\{\LCS(x_0,x_1)\ge\rho n\}$. On $G$, both possible choices of $I$
submit the same string $z^\ast$. Since the decoder is single-valued,
\[
\Pr[\Dec(\sk,z^\ast)=m_0]
+\Pr[\Dec(\sk,z^\ast)=m_1]\le1,
\]
where both probabilities use the decoder's fresh coins. Hence a uniform $I$
causes a decoding failure with conditional probability at least $1/2$, and the
adversary's winning probability is at least $\Pr[G]/2$. By
Lemma~\ref{lem:transfer}, $\Pr[G]\ge\beta_\rho-\negl(\lambda)$ for every block
length and $\Pr[G]\ge1-\negl(\lambda)$ when
$n=\omega(\log\lambda)$, proving both claims.
\end{proof}

\section{List decoding}\label{sec:stair}

We extend the collision construction to $L+1$ messages to defeat list-decoders of list
size $L$. The combinatorial input is the $m$-wise Chv\'atal--Sankoff constant. In this section, a list-decoding PRC has the same key generation, encoding,
  and pseudorandomness requirements as Definition~\ref{def:prc}, but
  \[
  \Dec(\sk,\cdot):\Sigmaq^\ast\to
  \{T\subseteq\bits^{k(\lambda)}:|T|\le L\}.
  \]

\begin{definition}[List decoder]\label{def:list}
An \emph{$L$-list decoder} is possibly randomised and returns a subset of the message
space, and satisfies $|\Dec(\sk,z;r)|\le L$ for every $z$ and every fixing $r$
of its coins. The code is \emph{robust} if, for every $E\in\Cpk$, there is a
negligible function $\mu_E$ such that, for every security parameter $\lambda$
and every message $m\in\bits^{k(\lambda)}$,
\[
\Pr_{\substack{(\sk,\pk)\leftarrow\KGen(1^\lambda),\ x\leftarrow\Enc(\pk,m)\\
E,\Dec\text{-coins}}}
\big[m\in\Dec(\sk,E(1^\lambda,\pk,x))\big]\ge1-\mu_E(\lambda).
\]
The empty list represents failure. The probability is over key generation,
encoding, the channel's coins, and decoder coins; there is no distribution over
messages, and the same negligible bound $\mu_E$ applies to every message at each
$\lambda$. List size $L=1$ is the single-valued case of
Definition~\ref{def:prc}, after identifying $\bot$ with the empty list.
\end{definition}

\begin{definition}[$m$-wise Chv\'atal--Sankoff constant]\label{def:gm}
For independent uniform $U^{(1)},\dots,U^{(m)}\in\bits^n$, let
$\gm{m}:=\lim_{n\to\infty}\E[\LCS(U^{(1)},\dots,U^{(m)})]/n$, when the limit
exists. Thus $\gm2=\gamma_2$.
\end{definition}

\begin{lemma}[Existence and monotonicity]\label{lem:gm-exist}
For every constant $m\ge2$, $\gm{m}$ exists and equals $\sup_n\E[\LCS_m]/n$;
moreover $\gm{m+1}\le\gm{m}\le\gamma_2$, and each finite-$n$ value
$\E[\LCS_m]/n$ is a lower bound on $\gm m$.
\end{lemma}
\begin{proof}
Write $\ell^{(m)}_n:=\E[\LCS(U^{(1)},\dots,U^{(m)})]$. Cutting all $m$ strings at
the same index $n_1$ into independent uniform blocks and concatenating a common
subsequence of the first blocks with one of the second blocks shows
$\LCS_m$ is superadditive pointwise, hence $\ell^{(m)}_{n_1+n_2}\ge
\ell^{(m)}_{n_1}+\ell^{(m)}_{n_2}$. As $0\le\ell^{(m)}_n\le n$, Fekete's lemma
gives $\ell^{(m)}_n/n\to\sup_t\ell^{(m)}_t/t=:\gm m$. Since every ratio
$\ell^{(m)}_n/n$ is at most this supremum, each finite-$n$ ratio is a lower
bound on $\gm m$. Any string common to
$m+1$ strings is common to any $m$ of them, so $\LCS_{m+1}\le\LCS_m$ pointwise,
giving $\gm{m+1}\le\gm m$; the base $\gm2=\gamma_2$ gives $\gm m\le\gamma_2$.
\end{proof}

\begin{lemma}[Multiple-string rate~\cite{LiRenWen}]\label{lem:gm-limit}
There are absolute constants $0<c_1<c_2$ such that, for every integer $m\ge2$,
\[
\frac12+\frac{c_1}{\sqrt m}\ \le\ \gm m\ \le\
\frac12+\frac{c_2}{\sqrt m}.
\]
In particular, $\gm m>1/2$ for every finite $m$ and $\gm m\to1/2$.
\end{lemma}
\begin{proof}
This is Theorem~1.1 of Li, Ren, and Wen~\cite{LiRenWen}, translated from their
notation $\gamma_{2,d}$ (with $d=m$) to $\gm m$.
\end{proof}


\begin{lemma}[$m$-wise canonical rule and concentration]\label{lem:mwise}
For every constant integer $m\ge2$ there is a deterministic,
permutation-invariant rule
$Z(x_1,\dots,x_m)$ that is a common subsequence of all inputs with
$|Z|=\LCS_m$, computable in $O(n^m)$ time. Given $Z$, deterministic witnessing
sets $J_i\subseteq[n]$ satisfying $(x_i[j])_{j\in J_i}=Z$ can additionally be
computed in $O(mn)$ time. Moreover, let $\mathcal C$ be an undetectable
binary public-key PRC, let $\mu_1,\dots,\mu_m$ be uniformly computable message families, sample $(\sk,\pk)\leftarrow\KGen(1^\lambda)$ and, conditional on $\pk$, independently draw $x_i\leftarrow\Enc(\pk,\mu_i)$. For every
fixed rational $0<\rho<\gm m$ there are a constant $N_{m,\rho}$ and a
negligible function $\mu_{m,\rho}^{\mathrm{undet}}$ such that, for every
$\lambda$ with $n(\lambda)\ge N_{m,\rho}$,
\[
\Pr\big[\LCS(x_1,\dots,x_m)\ge\rho n\big]
\ \ge\ 1-\exp\!\big(-\tfrac{(\gm m-\rho)^2n}{2m}\big)
-\mu_{m,\rho}^{\mathrm{undet}}(\lambda).
\]
For every block length, the same event has probability at least
$\beta_{m,\rho}-\negl(\lambda)$ for a constant $\beta_{m,\rho}>0$.
\end{lemma}

\begin{proof}
  Canonically sort the inputs using the fixed order $\prec$ and apply the
  standard $m$-dimensional LCS dynamic program with deterministic tie-breaking.
  The resulting longest common subsequence $Z$ depends only on the input
  multiset and is therefore permutation-invariant. The dynamic-programming table
  has $O(n^m)$ entries, and greedy leftmost embeddings of $Z$ in the original
  inputs produce the witnessing sets $J_i$ in $O(mn)$ time.

  Let
  \[
  \alpha:=\frac{\gm m-\rho}{2}>0.
  \]
  The value $\LCS_m$ changes by at most one when any one of the $mn$ input bits
  is changed. For all sufficiently large $n$,
  \[
  \E[\LCS_m]\ge(\gm m-\alpha)n=(\rho+\alpha)n.
  \]
  McDiarmid's inequality therefore gives
  \[
  \Pr[\LCS_m<\rho n]
  \le
  \exp\!\left(-\frac{2\alpha^2n^2}{mn}\right)
  =
  \exp\!\left(-\frac{(\gm m-\rho)^2n}{2m}\right).
  \]
  Choose $N_{m,\rho}$ large enough that this bound is valid and the corresponding
  success probability is at least $1/2$ whenever $n\ge N_{m,\rho}$. For smaller
  $n$, the event that all $m$ uniform strings are identical has probability
  $2^{-(m-1)n}$ and implies $\LCS_m=n$. Hence
  \[
  \beta_{m,\rho}
  :=
  \min\!\left\{\frac12,
  2^{-(m-1)(N_{m,\rho}-1)}\right\}>0
  \]
  is a uniform lower bound for every block length.

  Finally, write the fixed rational $\rho$ as $r/s$. Because $m$ is constant,
  the predicate
  \[
  \mathbf 1[s\LCS_m\ge rn]
  \]
  is computable in $O(n^m)$ time. An $m$-query distinguisher evaluates this
  predicate on independent encodings of the message families. Under the ideal
  oracle its inputs are independent uniform strings, so the bounds above apply;
  pseudorandomness transfers both bounds to the real codewords up to negligible
  error.
  \end{proof}

\begin{theorem}[List-size staircase]\label{thm:stair}
Let $L\ge1$ be a constant, and suppose $2^{k(\lambda)}\ge L+1$ for all
sufficiently large $\lambda$. For every constant deletion fraction
$\delta\in(1-\gm{L+1},1]$, no binary, undetectable public-key PRC that decodes
to lists of size at most $L$ is robust against $\Cpk$. More quantitatively, for
$m=L+1$ and any rational $1-\delta<\rho<\gm m$, the constructed pairs satisfy,
pointwise for all sufficiently large $\lambda$,
\[
\min_{i\in[m]}\Pr[m_i\in\Dec(\sk,C_i(1^\lambda,\pk,x_i))]
\ \le\ 1-\frac{\beta_{m,\rho}}{m}+\negl(\lambda).
\]
If $n=\omega(\log\lambda)$, the right-hand side sharpens to
$1-1/(L+1)+\negl(\lambda)$.
\end{theorem}
\begin{proof}
Set $m:=L+1$ and $S:=\Enc$, and choose rational constants
  \[
  1-\delta<\rho<\gm m,
  \qquad
  1-\rho<\widehat\delta<\delta.
  \]
  For $j\in[m]$, let $m_j(\lambda)$ be the $k(\lambda)$-bit binary
  representation of $j-1$ when $k(\lambda)\ge\lceil\log m\rceil$, and set
  $m_j(\lambda):=0^{k(\lambda)}$ otherwise. These families are uniformly
  PPT-computable and are distinct for all sufficiently large $\lambda$.

  For each $i\in[m]$, define $C_i(1^\lambda,\pk,y)$ as follows. Let $N:=|y|$.
  If $k(\lambda)<\lceil\log m\rceil$, return the prefix of $y$ of length
  $\lceil(1-\widehat\delta)N\rceil$. Otherwise, place $y$ in coordinate $i$ and
  independently sample $x_j\leftarrow S(\pk,m_j)$ for every $j\ne i$. If all
  samples have length $N$, compute
  \[
  z:=Z(x_1,\dots,x_{i-1},y,x_{i+1},\dots,x_m)
  \]
  and its witnessing embedding $J_y$ in $y$. Return $(y_j)_{j\in J_y}$ if
  $|z|\ge(1-\widehat\delta)N$; in all other cases, return the same prefix.
  Thus each $C_i$ is PPT, accesses no key material beyond $\pk$, and always
  outputs a subsequence obtained using at most
  $\widehat\delta N\le\delta N$ deletions. Hence $C_i\in\Cpk$.

Fix a sufficiently large $\lambda$. Sample
  $(\sk,\pk)\leftarrow\KGen(1^\lambda)$ and, conditional on $\pk$,
  independently sample
  \[
  x_j\leftarrow\Enc(\pk,m_j),
  \qquad j\in[m].
  \]
  Couple the invocation of $C_i$ on $x_i$ by using $x_j$ as its internal sample
  in coordinate $j$ for every $j\ne i$. Each coupled invocation then has the
  same distribution as its standalone execution.

  Let
  \[
  z^\ast:=Z(x_1,\dots,x_m),
  \qquad
  G:=\{\LCS(x_1,\dots,x_m)\ge\rho n\},
  \qquad
  \eta:=\Pr[\overline G].
  \]
  On $G$, every channel takes the in-budget branch and, by
  permutation-invariance of $Z$, outputs $z^\ast$. Lemma~\ref{lem:mwise} gives
  \[
  \eta\le1-\beta_{m,\rho}+\negl(\lambda).
  \]
  If $n=\omega(\log\lambda)$, it instead gives
  \[
  \eta\le
  \exp\!\left(-\frac{(\gm m-\rho)^2n}{2m}\right)
  +\negl(\lambda)
  =\negl(\lambda).
  \]

  Share one uniform decoder tape among the $m$ coupled invocations, and let
  \[
  p_i:=\Pr[m_i\in\Dec(\sk,C_i(1^\lambda,\pk,x_i))]
  \]
  be the corresponding standalone success probability. On $G$, the common
  output list contains at most $L=m-1$ of the messages; outside $G$, there are
  at most $m$ successes. Therefore
  \[
  \sum_{i=1}^m p_i
  \le L\Pr[G]+m\Pr[\overline G]
  =L+\eta,
  \]
  and hence
  \[
  \min_i p_i
  \le\frac{L+\eta}{m}
  \le1-\frac{\beta_{m,\rho}}m+\negl(\lambda).
  \]
  Since there are finitely many indices, one fixed channel--message pair
  satisfies this bound for infinitely many $\lambda$, contradicting robustness.
  When $n=\omega(\log\lambda)$, using $\eta=\negl(\lambda)$ gives
  \[
  \min_i p_i
  \le1-\frac1m+\negl(\lambda)
  =1-\frac1{L+1}+\negl(\lambda).
  \]
\end{proof}


\section{Scope of the collision argument}\label{sec:tight}

  The collision argument requires distinct messages, uses public-key-dependent resampling, and weakens as the alphabet grows.

  When $k=0$, the two prospective decoding events concern the same message and
  therefore do not conflict and thus Lemma~\ref{lem:collision} gives no zero-bit
  impossibility. This is consistent with known zero-bit binary PRCs for random
  deletions~\cite{ChristGunn}.

  For comparison, when the channel tapes are independent of the encoder tapes
  the two channel outputs agree with probability at most
  \[
  \exp\!\bigl(-2n(1/2-\delta)^2\bigr)+\negl(\lambda)
  \]
  for $\delta<1/2$; see Appendix~\ref{app:independent-executions}. This does not
  constrain the marginal-preserving coupling used in
  Lemma~\ref{lem:collision}.

 Two boundary cases can be settled directly. Short block lengths are ruled out already by the identity channel, and above deletion fraction one half even key-independent channels suffice.

\begin{proposition}[Logarithmic-block baseline]\label{prop:short-block}
Let $\mathcal C=(\KGen,\Enc,\Dec)$ be a binary, multi-bit public-key PRC with a
single-valued decoder and the repeated-query pseudorandomness of
Definition~\ref{def:undet}. If $n(\lambda)=O(\log\lambda)$, then $\mathcal C$
cannot satisfy the robustness requirement even for the identity channel.
\end{proposition}
\begin{proof}
Suppose otherwise. Let $m_0:=0^{k(\lambda)}$ and
$m_1:=0^{k(\lambda)-1}1$, and let $P_{i,\pk}$ be the conditional distribution
of $\Enc(\pk,m_i)$. Robustness against the identity channel and
single-valuedness imply
\[
\E_{\pk}\!\left[
d_{\mathrm{TV}}(P_{0,\pk},P_{1,\pk})
\right]\ge1-2\varepsilon(\lambda)
\]
for a negligible function $\varepsilon$. Consequently,
\[
\Pr_{\pk}\!\left[
d_{\mathrm{TV}}(P_{0,\pk},P_{1,\pk})<\tfrac12
\right]\le4\varepsilon(\lambda).
\]

Since $n=O(\log\lambda)$, the output space has size
$M:=2^n=\poly(\lambda)$. A PPT distinguisher can make $\poly(M)$ repeated
queries on each message and estimate the total variation distance between
the two empirical encoding distributions to constant accuracy. In the real
experiment this distance is at least $1/2$ except with negligible
probability, whereas in the ideal experiment both messages produce the same
uniform distribution. This contradicts Definition~\ref{def:undet}.
\end{proof}

  \begin{proposition}[Key-independent threshold at one half]
  \label{prop:majority}
  For every constant $\delta\in[1/2,1]$, no binary, multi-bit, undetectable
  public-key PRC with a single-valued decoder is robust against $\Cpk$, even
  when restricted to key-independent channels.
  \end{proposition}

  \begin{proof}
  For $y\in\bits^N$, let $\sigma(y)$ be its majority symbol, resolving ties in favour of
  $0$, and define
  \[
  C_{\mathrm{maj}}(1^\lambda,\pk,y)
  :=
  \sigma(y)^{\lceil N/2\rceil}.
  \]
  This deterministic channel ignores $\pk$ and deletes at most
  $\lfloor N/2\rfloor\le\delta N$ symbols, so
  $C_{\mathrm{maj}}\in\Cpk$.

  Consider the canonical messages $m_0=0^k$ and $m_1=0^{k-1}1$, and
  independently sample $x_i\leftarrow\Enc(\pk,m_i)$. For independent uniform
  $U,U'\in\bits^n$, writing $s:=\Pr[\sigma(U)=0]$ gives
  \[
  \Pr[\sigma(U)=\sigma(U')]
  =
  s^2+(1-s)^2
  \ge\frac12.
  \]
  A two-query application of pseudorandomness therefore yields
  \[
  \eta
  :=
  \Pr[\sigma(x_0)\ne\sigma(x_1)]
  \le\frac12+\negl(\lambda).
  \]

  Share one uniform decoder tape between the two prospective decoding events.
  When $\sigma(x_0)=\sigma(x_1)$, both channel outputs coincide, so at most one
  message is decoded correctly; otherwise there are at most two successes.
  Hence, for the standalone success probabilities $p_0,p_1$,
  \[
  p_0+p_1\le1+\eta\le\frac32+\negl(\lambda),
  \]
  and therefore
  \[
  \min\{p_0,p_1\}\le\frac34+\negl(\lambda).
  \]
  This contradicts robustness.
  \end{proof}

  For each fixed alphabet size $q$, Corollary~\ref{cor:qary} gives the ceiling
  \[
  \delta_c(q):=1-\gamma_q^{\mathrm{LCS}}.
  \]
  Since Kiwi, Loebl, and Matou\v{s}ek~\cite{KLM} prove
  \[
  \gamma_q^{\mathrm{LCS}}\sim\frac{2}{\sqrt q},
  \qquad
  \delta_c(q)=1-\frac{2+o(1)}{\sqrt q}\longrightarrow1,
  \]
  the pairwise-LCS method becomes weaker as $q$ grows. The corollary assumes
  that $q$ is fixed independently of $\lambda$ and does not cover alphabets
  growing with the security parameter, such as those of Golowich and
  Moitra~\cite{GolowichMoitra}.
\section{Conclusion and open problems}\label{sec:conclusion}
We prove that, for every fixed alphabet size $q\ge2$, no multi-bit,
  undetectable public-key PRC with a single-valued decoder is robust against
  public-key-dependent deletions for any constant
  \[
  \delta>1-\gamma_q^{\mathrm{LCS}}.
  \]
  For $q=2$ current bounds rule out every $\delta>0.207334008$ and for every fixed list size $L$, provided the message space contains at least $L+1$ elements for all sufficiently large $\lambda$, the same collision argument gives the threshold
  $1-\gm{L+1}$, which approaches $1/2$ as
  $1/2-\Theta(1/\sqrt L)$. The proof uses pseudorandomness to transfer
  random-string LCS behaviour to codewords and public samplability to force
  distinct messages onto a common subsequence.

Open problems include determining the deletion threshold for key-independent
  channels in the unresolved interval $(1-\gamma_2,1/2)$, sharpening finite-$m$
  bounds on $\gm m$, and establishing deletion bounds under other natural channel
  models.
\section*{AI Disclosure}
We employed ChatGPT/Codex 5.5 (OpenAI) to assist with wording and proof checking. It was also used in the ideation and
discussion of the list-decoding extension (section~\ref{sec:stair}). We take full responsibility for the correctness, originality, and accuracy of the manuscript.

\bibliographystyle{alpha}
{\footnotesize
\bibliography{main}
}

\appendix
\section{Independent-execution comparison}
\label{app:independent-executions}

\begin{lemma}[Committed-string embedding]\label{lem:embed}
Fix any string $z\in\bits^{\ell}$ and let $x\in\bits^{n}$ be uniform. Then
\[
\Pr_x[z\subseq x]\ =\ \Pr\big[\Bin(n,\tfrac12)\ge \ell\big],
\]
a quantity depending only on $\ell$, not on the bits of $z$. In particular, for
$\ell>n/2$,
\[
\Pr_x[z\subseq x]\ \le\ \exp\!\Big(-2n\big(\tfrac{\ell}{n}-\tfrac12\big)^2\Big).
\]
\end{lemma}
\begin{proof}
Greedy left-to-right matching is optimal for testing $z\subseq x$: scan $x$ and
match the next unmatched symbol of $z$ whenever it appears. Each scanned bit of
$x$ equals the current target symbol of $z$ with probability $\tfrac12$,
independent of the target's value (as $x$ is uniform), so the number of $x$-bits
consumed to match all $\ell$ symbols is a sum of $\ell$ i.i.d.\ $\Geo(\tfrac12)$
variables. Thus $z\subseq x$ iff this sum is $\le n$, i.e.\ iff among $n$ i.i.d.\
fair coin flips at least $\ell$ are ``successes'' (negative-binomial/binomial
duality): $\Pr[z\subseq x]=\Pr[\Bin(n,\tfrac12)\ge \ell]$, independent of the bits
of $z$. Hoeffding's inequality gives the stated tail for $\ell>n/2$.
\end{proof}

Call a channel $E\in\Cpk$ \emph{key-independent} if
$E(1^\lambda,\pk,y)$ ignores the $\pk$ argument (and hence its output
distribution on $y$ is identical for every public key). It may depend on the
public parameter $1^\lambda$. Thus key-independent channels retain
the deletion-only, sure-budget, and PPT requirements of
Definition~\ref{def:class} on every well-formed input.

\begin{theorem}[Independent-execution collision bound]\label{thm:barrier}
Let $\mathcal C$ be an undetectable binary public-key PRC, let $m_0,m_1$ be
uniformly computable message families, and set
$x_i\leftarrow\Enc(\pk,m_i)$ independently under
$(\sk,\pk)\leftarrow\KGen(1^\lambda)$. Fix any constant
$\delta\in[0,1/2)$. 
For every pair $(C_0,C_1)$ of members of $\Cpk$, run the channels with random tapes. Conditional on $\pk$, each tape has its usual marginal distribution, and the pair of tapes is independent of the encoder tapes producing $x_0$ and $x_1$. The two channel tapes may otherwise be arbitrarily coupled.
\[
\Pr[C_0(1^\lambda,\pk,x_0)=C_1(1^\lambda,\pk,x_1)]
\ \le\ \exp\!\big(-2n(\tfrac12-\delta)^2\big)+\negl(\lambda).
\]
When $n=\omega(\log\lambda)$ this bound is negligible: in that regime, such
executions whose channel tapes are independent of the encoder tapes produce equal outputs only with negligible
probability for every constant $\delta\in[0,\tfrac12)$, including throughout
$(0.207334008,1/2)$.
\end{theorem}
\begin{proof}
Let $y_0:=C_0(1^\lambda,\pk,x_0)$. On the collision event,
$y_0=C_1(1^\lambda,\pk,x_1)\subseq x_1$ and
$|y_0|\ge(1-\delta)n$; hence the collision event is contained in
$B:=\{y_0\subseq x_1\}$.

Among the channel tapes, the event $B$ depends only on the marginal execution of $C_0$. Hence the argument applies to any coupling of the two channel tapes that preserves their usual marginals and is independent of the encoder randomness.

Now use a two-query distinguisher given $(1^\lambda,\pk)$. It samples $C_0$'s random tape with its usual distribution, independently of its oracle answers, computes
$m_0,m_1$ from $1^\lambda$, queries $m_0$, receives $a$, computes
$y_0:=C_0(1^\lambda,\pk,a)$, queries $m_1$, receives $b$, and outputs
$\mathbf1[y_0\subseq b]$. Under the ideal oracle, $b$ is uniform and independent
of the jointly generated tuple $(\pk,a,y_0)$. Pointwise for every realised $y_0$,
Lemma~\ref{lem:embed} and $|y_0|\ge(1-\delta)n$ give
\[
\Pr[y_0\subseq b\mid y_0]
 \le\exp\!\big(-2n(\tfrac12-\delta)^2\big).
\]
Averaging proves the same ideal-world bound. Pseudorandomness transfers the
event $B$ to the real two-query experiment up to $\negl(\lambda)$, and the
containment of the collision event in $B$ proves the claim.
\end{proof}

\end{document}